\documentclass[12pt]{scrartcl}

\usepackage{fullpage}
\usepackage[latin1]{inputenc}
\usepackage[english]{babel}
\usepackage{latexsym}
\usepackage{amsmath}
\usepackage{amsthm}
\usepackage{amssymb}
\usepackage{tikz}
\usetikzlibrary{patterns}
\usepackage{xcolor}
\usetikzlibrary{shapes.misc}
\usetikzlibrary{decorations.pathreplacing,calligraphy}
\usepackage{enumerate}

\colorlet{ColOpen}{black!40}
\colorlet{ColTrue}{black}
\colorlet{ColFail}{black!18}

\newtheorem{theorem}{Theorem}

\newtheorem*{conjecture}{Conjecture}

\newtheorem{remark}[theorem]{Remark}

\newcommand{\DD}[2][]{\frac{\mathrm{d}^2{#1}}{\mathrm{d}{#2}^2}}

\newcommand{\norm}[2][]{\left\|#2\right\|_{#1}}

\newcommand{\N}{\mathbb{N}}

\newcommand{\R}{\mathbb{R}}
\newcommand{\C}{\mathbb{C}}

\newcommand{\id}{\mathbb{I}}

\newcommand{\eps}{\varepsilon}
\newcommand{\ii}{\mathrm{i}}
\newcommand{\e}{\mathrm{e}}

\newcommand{\Hs}{\mathcal{H}}

\newcommand{\tr}{\mathrm{tr}} 
\newcommand{\dd}{\, \mathrm{d}}

\begin{document}

\title{The state of the Lieb--Thirring conjecture}
\date{}
\author{Lukas Schimmer}
\publishers{\small Institut Mittag-Leffler, The Royal Swedish Academy of Sciences, 182 60 Djursholm, Sweden, lukas.schimmer@kva.se}
 \maketitle
\begin{abstract}
In 1976 Lieb and Thirring established upper bounds on sums of powers of the negative eigenvalues of a Schr{\"o}dinger operator in terms of semiclassical phase-space integrals. Over the last 45 years the optimal constants in these inequalities, the values of which were conjectured by Lieb and Thirring,  have been subject of intense investigations. We aim to review existing results.
\end{abstract}

\emph{This paper is dedicated to Elliott H.~Lieb on the occasion of his 90th birthday with gratitude for his mentorship and for all that he has taught me.}

\section{Introduction}\label{sec:intro}
Estimates on the bound states of a Hamiltonian are of importance in quantum mechanics. For a Schr{\"o}dinger operator $-\Delta+V$ on $L^2(\R^d)$ with decaying potential $V:\R^d\to\R$ of particular interest are bounds on power series $\sum_{j\ge 1}|E_j|^\gamma$ of the negative eigenvalues $E_1\le E_2\le\dots<0$.  
Illuminating heuristics can be obtained by recalling the semiclassical approach, which dates back to the early days of quantum mechanics. Each volume $(2\pi)^d$ in $2d$-dimensional phase space is assumed to support one quantum state. Identifying the Hamiltonian with the classical Hamilton function $p^2+V(x)$ one can thus hope to approximate  $\sum_{j\ge 1}|E_j|^\gamma$ by
\begin{align*}
\sum_{j\ge1}|E_j|^\gamma\approx\frac{1}{(2\pi)^d}\int_{\R^d}\int_{\R^d}(p^2+V(x))_-^\gamma\dd p\dd x\,.
\end{align*}
Here (and in the remainder) $a_-=(|a|-a)/2$ denotes the negative part of a variable, function or self-adjoint operator.
The integration over the momentum $p$ can be carried out explicitly, leading to  
\begin{align*}
\sum_{j\ge1}|E_j|^\gamma\approx L_{\gamma,d}^{\mathrm{cl}}\int_{\R^d}V(x)_-^{\gamma+d/2}\dd x
\end{align*}
with the constant 
\begin{align*}
L_{\gamma,d}^{\mathrm{cl}}=\frac{1}{(2\pi)^d}\int_{\R^d}(1-p^2)_+^\gamma\dd p
=\frac{\Gamma(\gamma+1)}{(4\pi)^{d/2}\Gamma(\gamma+\frac{d}{2}+1)}\,.
\end{align*}
Although these arguments are not mathematically rigorous, the approximation above is true in the limit of large potentials. To be more precise, if $E_j(\lambda)$ are the negative eigenvalues of $-\Delta+\lambda V$ then for all $\gamma\ge0$ the so-called Weyl asymptotics hold
\begin{align}
\lim_{\lambda\to \infty}\frac{\sum_{j\ge 1}|E_j(\lambda)|^\gamma}{\lambda^{\gamma+d/2}}
=L_{\gamma,d}^{\mathrm{cl}}\int_{\R^d}V(x)_-^{\gamma+d/2}\dd x
\label{eq:Weyl} 
\end{align}
 under suitable assumptions on $V$. An equivalent asymptotic result can be obtained by considering $-\hbar^2\Delta+V$ and letting Planck's constant $\hbar\to 0$. In 1976 Lieb and Thirring \cite{Lieb1976} proved that the semiclassical approximation also holds as an upper bound with a possibly larger constant $L_{\gamma,d}$, establishing what is now referred to as Lieb--Thirring inequalities. 
\begin{theorem}[Lieb--Thirring inequalities]\label{th:LT}
Let $d\ge1$ and $\gamma\ge0$ with 
\begin{align*}
\gamma\ge1/2&\qquad\text{if}\qquad d=1\,,\\
\gamma>0&\qquad\text{if}\qquad d=2\,,\\
\gamma\ge 0&\qquad\text{if}\qquad d\ge 3\,.
\end{align*}
Then there is a constant $L_{\gamma,d}$ such that for all real-valued $V\in L^{\gamma+d/2}(\R^d)$ the negative eigenvalues $E_j$ of $-\Delta+V$ satisfy
\begin{align}
\sum_{j\ge1}|E_j|^\gamma\le L_{\gamma,d}\int_{\R^d}V(x)_-^{\gamma+\frac d2}\dd x\,.
\label{eq:LT}
\end{align}
\end{theorem}
The inequalities were all proved in Lieb and Thirring's original work \cite{Lieb1976} with the exception of the endpoint cases $d=1,\gamma=1/2$ and $d\ge3,\gamma=0$. The former was established by Weidl \cite{Weidl1996} and the latter by Cwikel \cite{Cwikel1977}, Lieb \cite{Lieb1976b,Lieb1980} and Rozenblum \cite{Rozenblum1972}. They will be discussed in Sections \ref{sec:WHLT} and \ref{sec:CLR}, respectively. 
\begin{remark}\label{rem:sc}
To emphasise that the inequality provides an upper bound on the quantum mechanical quantity  $\sum_{j\ge1}|E_j|^\gamma$ in terms of the semiclassical quantity $\int_{\R^d}V_-^{\gamma+d/2}\dd x$ arising in a phase-space computation, it is illustrative to write it as
\begin{align*}
\tr_{L^2(\R^d)}(-\Delta+V)_-^\gamma
\le \frac{R_{\gamma,d}}{(2\pi)^d}\int_{\R^d}\int_{\R^d}(p^2+V(x))_-^\gamma\dd p\dd x
\end{align*}
with the constant $R_{\gamma,d}=L_{\gamma,d}/L_{\gamma,d}^{\mathrm{cl}}$. 
\end{remark}
In contrast to Weyl asymptotics, the upper bounds \eqref{eq:LT} cannot hold for all $\gamma\ge0$. The restrictions presented above are optimal. In dimensions $d=1,2$ even weak potentials can support negative eigenvalues due to $0$ being a resonance state. More precisely, for any potential $V\in\mathcal{C}_0^\infty(\R^d)$ with $\int_{\R^d} V\dd x<0$ the Schr{\"o}dinger operator already has at least one negative eigenvalue such that \eqref{eq:LT} fails for $\gamma=0$ in dimensions $d=1,2$. Furthermore, this eigenvalue is of order $\lambda^2$ if  $V$ is replaced by $\lambda V$  in one dimension. Letting $\lambda\to 0$ shows that \eqref{eq:LT} cannot hold if $d=1$ and $2\gamma<\gamma+1/2$. 

While Lieb and Thirring's proof \cite{Lieb1976} did not yield the smallest possible $L_{\gamma,d}$ in \eqref{eq:LT}, the authors conjectured the optimal values of these constants. Over the last 45 years there has been tremendous interest in investigating their conjecture. This review aims to focus on the current state of the subject matter. For a comprehensive review of Lieb--Thirring inequalities including recent research directions we refer to Frank's \cite{Frank2020b}.  We also mention Hundertmark's review \cite{Hundertmark2007} which similarly provides a summary of the state of the field in 2007. A pedagogical introduction to the subject can be found in an upcoming book by Frank, Laptev and Weidl \cite{Frank2021}.

The inequalities have  been extended in various ways. While several of these extensions will be mentioned below in the context of  Lieb and Thirring's conjecture, a complete survey is beyond the scope of this article and some exciting results will have to be omitted. This includes in particular Lieb--Thirring inequalities for complex potentials as well as for fractional Laplacians and Pauli operators. More details can be found in e.g.~\cite{Frank2020b}.

The most famous application of Lieb--Thirring inequalities concerns the stability of matter problem. The relevant dual kinetic energy inequalities will be discussed in Section \ref{sec:kinetic}. However, the inequalities have found applications far beyond that, including in questions relating to the absolute continuity of the spectrum of Schr{\"o}dinger operators and to  Navier--Stokes equations. We refer to \cite{Laptev2012} by Laptev and \cite{Ilyin2022} by Ilyin, Kostianko, and  Zelik, as well as references therein.

We conclude this section by giving the main ideas that Lieb and Thirring used to prove their result.

\begin{proof}[Proof of Theorem \ref{th:LT}]
It is sufficient to consider $V\le0$ and $V\in\mathcal{C}_0^\infty(\R^d)$.  The proof relies on establishing a bound on $N_E(V)$, the number of eigenvalues of $-\Delta+V$ below $E<0$, using the fact that
\begin{align}
\sum_{j\ge1} |E_j|^\gamma=\gamma\int_0^\infty e^{\gamma-1} N_{-e}(V) \dd e\,.
\label{eq:gammaNE}
\end{align}
To study $N_E$, one observes that whenever $E$ is a negative eigenvalue of $H=-\Delta+V$ with eigenfunction $\psi$, then (at least formally) $\sqrt{|V|}\psi$ satisfies $K_E\sqrt{|V|}\psi=\sqrt{|V|}\psi$ with
\begin{align*}
K_E= \sqrt{|V|}(-\Delta-E)^{-1}\sqrt{|V|}\,.
\end{align*}
The operator $K_E$, referred to as the Birman--Schwinger operator, was introduced independently by Birman \cite{Birman1961} and Schwinger \cite{Schwinger1961}. It has several nice properties that facilitate the study of its spectrum, in particular
\begin{enumerate}[(i)]
\item $K_E\ge0$ is compact with eigenvalues $\mu_1(K_E)\ge\mu_2(K_E)\ge\dots\ge0$,
\item $\mu_j(K_{E_j})=1$ \emph{(Birman--Schwinger principle)},
\item $K_E< K_{E'}$ for $E< E'<0$ \emph{(monotonicity)},
\end{enumerate}
where (ii) formalises the aforementioned connection between the spectra of $-\Delta+V$ and $K_E$. Using (i)-(iii) together with a continuity argument, one can prove that $N_E(V)$ coincides with he number of eigenvalues of $K_E$ above $1$, which implies that $N_E(V)\le \tr_{L^2(\R^d)} K_E^m$ for any $m\ge 1$. Upper bounds on such traces can be obtained by means of the abstract inequality \begin{align*}
\tr(B^{1/2}AB^{1/2})^m\le \tr (B^{m/2}A^mB^{m/2})\quad\text{for}\quad A,B\ge0,
\end{align*}
which was established in \cite{Lieb1976} and is likewise referred to as Lieb--Thirring inequality.

Using that $(-\Delta-E)^{-m}$ is given by $((2\pi p)^2-E)^{-m}$ in Fourier space, the diagonal of its integral kernel $G_{E,m}(x-y)$ can be computed to be $G_{E,m}(0)=c_{m,d} |E|^{-m+d/2}$ with some $c_{m,d}>0$. Thus, from the arguments above
\begin{align*}
N_E(V)\le \tr_{L^2(\R^d)} K_E^m\le c_{m,d} |E|^{-m+d/2}\int_{\R^d} |V(x)|^m\dd x\,.
\end{align*}
Combining this with \eqref{eq:gammaNE} would yield an integral that diverges in $e=-E$. However, on account of the variational principle $N_E(V)\le N_{E/2}(-(V-E/2)_-)$  and thus
\begin{align*}
\sum_{j\ge1} |E_j|^\gamma \le 2^{m-d/2}c_{m,d}\gamma\int_0^\infty\int_{\R^d} e^{\gamma-1-m+d/2} (V(x)+e/2)_-^m\dd x\dd e\,.
\end{align*}
If $m$ is chosen to satisfy $d/2<m<\gamma+d/2$ the integral is finite and of the desired form, proving Theorem \ref{th:LT} in all but the endpoint cases.
\end{proof}


\section{Dual kinetic energy inequality}\label{sec:kinetic}
Lieb and Thirring originally introduced \eqref{eq:LT} to prove stability of matter in quantum mechanics \cite{Lieb1975}. Their proof crucially relied on a lower bound on the kinetic energy of electrons. The desired inequality is equivalent to \eqref{eq:LT} for $\gamma=1$. We will state the bound in three versions, each slightly more general than the previous one. 

\begin{theorem}[Kinetic energy inequality for orthonormal functions]\label{th:kinetic}
Let $d\ge1$. There is a constant $K_d$ such that for all $N\in\N$ and all orthonormal $\psi_1,\dots,\psi_N\in L^2(\R^{d})$ with $\nabla \psi_j\in L^2(\R^d)$
\begin{align}
\sum_{j=1}^N\int_{\R^d}|(\nabla \psi_j)(x)|^2\dd x\ge K_d\int_{\R^d}\left(\sum_{j=1}^N|\psi_j(x)|^2\right)^{1+2/d}\dd x\,.
\label{eq:kinetic} 
\end{align}  
The optimal constant $K_d$ is related to the optimal constant $L_{1,d}$ in \eqref{eq:LT} via
\begin{align}
\big((1+{d}/{2})L_{1,d}\big)^{1+{2}/{d}}\big((1+{2}/{d})K_{d}\big)^{1+{d}/{2}}=1\,.
\label{eq:LK}
\end{align}
\end{theorem}

\begin{proof}
Let $V\in L^{1+d/2}(\R^d)$ and let $E_j$ denote the negative eigenvalues of $-\Delta+V$. By the variational principle for sums of eigenvalues
\begin{align}
\sum_{j=1}^N\int_{\R^d}\left(|(\nabla \psi_j)(x)|^2+V(x)|\psi_j(x)|^2\right)\dd x\ge \sum_{j=1}^N E_j
\label{eq:sumvar}
\end{align}
for any orthonormal $\psi_1,\dots,\psi_N$ with the properties stated above. Denoting $\rho(x)=\sum_{j=1}^N|\psi_j(x)|^2$ and applying the Lieb--Thirring inequality yields the lower bound
\begin{align*}
\sum_{j=1}^N\int_{\R^d}|(\nabla \psi_j)(x)|^2\dd x
&\ge -L_{1,d}\int_{\R^d}V(x)_-^{1+d/2}\dd x-\int_{\R^d}V(x)\rho(x)\dd x
\end{align*}
The right side can be maximised over the choice of  $V\in L^{1+d/2}(\R^d)$. Clearly the optimal choice of $V$ is of the form $-c\rho(x)^{2/d}$ with $c>0$. This already yields a bound of the form \eqref{eq:kinetic} and  computing the optimal constant $c$ to be
\begin{align*}
c=\left((1+d/2)L_{1,d}\right)^{-2/d}
\end{align*}
once can conclude that
\begin{align*}
K_d\ge(1+2/d)^{-1}\left((1+d/2)L_{1,d}\right)^{-2/d}\,.
\end{align*}

Conversely, let $V\in L^{1+d/2}(\R^d)$ and let $E_j$ be the negative eigenvalues of the Schr{\"o}dinger operator $-\Delta+V$ with corresponding orthonormal eigenfunctions $\psi_j$. For fixed $N\ge1$ denote $\rho(x)=\sum_{j=1}^N|\psi_j(x)|^2$ and apply \eqref{eq:kinetic} to obtain
\begin{align*}
\sum_{j=1}^N|E_j|
&=-\sum_{j=1}^N\int_{\R^d}\left(|(\nabla \psi_j)(x)|^2+V(x)|\psi_j(x)|^2\right)\dd x\\
&\le -K_d\int_{\R^d}\rho(x)^{1+2/d}\dd x
+\int_{\R^d}V(x)_-\rho(x)\dd x\,.
\end{align*}
An upper bound can be obtained by maximising the right side with respect to all $\rho(x)\ge0$. The optimal choice is of the form $cV(x)_-^{d/2}$ with $c>0$. This already yields a bound of the form \eqref{eq:LT} with $\gamma=1$ and  computing the optimal constant $c$ to be
\begin{align*}
c=\left((1+2/d)L_{1,d}\right)^{-d/2}
\end{align*}
one can conclude that
\begin{align*}
L_{1,d}&\le(1+d/2)^{-1}\left((1+2/d)K_{1}\right)^{-d/2}\,. \qedhere
\end{align*}
\end{proof}

Inequality \eqref{eq:kinetic} should be compared to the Gagliardo--Nirenberg inequality
\begin{align}
\int_{\R^d}|(\nabla \psi)(x)|^2\dd x\ge K_d^{\mathrm{one}}\int_{\R^d}|\psi(x)|^{2+4/d}\dd x\,.
\label{eq:GN}
\end{align}
for normalised $\psi\in L^2(\R^d)$ with  $\nabla \psi\in L^2(\R^d)$.  By the same argument as in the proof of Theorem \ref{th:kinetic}, the optimal constant $K_{d}^{\mathrm{one}}$ in \eqref{eq:GN} and the optimal constant $L_{1,d}^{\mathrm{one}}$ in the bound
\begin{align*}
|E_1|\le L_{1,d}^{\mathrm{one}}\int_{\R^d} V(x)_-^{1+d/2}\dd x
\end{align*}
for the lowest eigenvalue of a Schr{\"o}dinger operator $-\Delta+V$ are related via the analogue of \eqref{eq:LK}. Clearly the optimal values satisfy $L_{1,d}\ge L_{1,d}^{\mathrm{one}}$ and equivalently $K_{d}\le K_{d}^{\mathrm{one}}$. 
On account of \eqref{eq:GN}, without the orthogonality condition in Theorem \ref{th:kinetic}, the constant $K_d$ on the right side of \eqref{eq:kinetic} has to be replaced by  $N^{-2/d}K_d^{\mathrm{one}}$. The absence of the factor $N^{-2/d}$   under the assumption of orthogonality is a manifestation of the Pauli exclusion principle.

To elaborate on this, we recall that the state of a quantum mechanical system of $N$ particles in $\R^d$ is (ignoring spin for the moment) described  by a normalised wave function $\psi\in L^2(\R^{dN})$. If the particles are fermions, $\psi$ is further required to be antisymmetric
\begin{align*}
\psi(x_1,\dots, x_i,\dots,x_j,\dots,x_N)=
-\psi(x_1,\dots, x_j,\dots,x_i,\dots,x_N)\quad\text{for}\quad i\neq j.
 \end{align*} 
This property is also referred to as fermions obeying the Pauli exclusion principle. 
The one-body density of $\psi$ is defined as 
\begin{align*}
\rho_\psi(x)=N\int_{\R^{d(N-1)}}|\psi(x,x_2,\dots,x_N)|^2\dd x_2\dots\dd x_N\,.
\end{align*}
Inequality \eqref{eq:kinetic} can be extended to this setting.

\begin{theorem}[Kinetic energy inequality for antisymmetric functions]\label{th:kineticas}
Let $d\ge1$. There is a constant $K_d$ such that for all $N\in\N$ and all antisymmetric, normalised $\psi\in L^2(\R^{dN})$ \begin{align}
\sum_{j=1}^N\int_{\R^{dN}}|(\nabla_j\psi)(x_1,\dots,x_N)|^2\dd x_1\dots\dd x_N\ge
 K_d\int_{\R^d}\rho_\psi(x)^{1+2/d}\dd x\,.
 \label{eq:kineticas}
\end{align}
The optimal constant $K_d$ is related to the optimal constant $L_{1,d}$ in \eqref{eq:LT} via \eqref{eq:LK}.
\end{theorem}
In the special case where $\psi(x_1,\dots,x_N)=\mathrm{det}(\psi_i(x_j))_{i,j=1}^N$ is a Slater determinant the result reduces to \eqref{eq:kinetic}. 
The proof of Theorem \ref{th:kineticas} is the same as the proof of Theorem \ref{th:kinetic}. The analogous inequality to \eqref{eq:sumvar},  
\begin{align*}
\sum_{j=1}^N\int_{\R^{dN}}\left(|(\nabla_j\psi)(x_1.\dots,x_N)|^2+V(x_j)|\psi(x_1,\dots,x_N)|^2\right)\dd x_1\dots\dd x_N\ge\sum_{j=1}^N E_j\,,
\end{align*}
 is a consequence of the Pauli exclusion principle, as the lowest possible energy of  $N$ noninteracting fermions in the presence of an electric potential $V$ is achieved by filling up the $N$ lowest energy levels of the one-particle Hamiltonian $-\Delta+V$. If the particles under consideration are bosons, the wave function is symmetric and inequality \eqref{eq:kineticas} only holds with $K_d$ replaced by $N^{-2/d}K_d^{\mathrm{one}}$. The fact that electrons are fermions and the resulting validity of \eqref{eq:kineticas} without the factor $N^{-2/d}$ were crucial in Lieb and Thirring's proof of the stability of matter in quantum mechanics  \cite{Lieb1975}. We also refer to \cite{Lieb2010} for a pedagogical introduction to the subject. 
If one considers particles with $q$ spin states, the one-body density has to be replaced by the spin-summed one-body density and the constant then has be decreased by a factor of $q^{-{2/d}}$. 

\begin{remark}[Kinetic energy inequality for density matrices]\label{rem:kineticg}
Yet another slight generalisation can be achieved by extending \eqref{eq:kineticas} to one-body density matrices. For any operator $\gamma$ on $L^2(\R^{d})$ with $0\le \gamma\le \id$ and $\tr_{L^2(\R^d)}\gamma<\infty$ the inequality then states that
\begin{align}
\tr_{L^2(\R^d)}(-\Delta\gamma)\ge K_d\int_{\R^d}\gamma(x,x)^{1+2/d}\dd x\,,
\label{eq:kineticg}
\end{align}
where $\gamma(x,x)$ denotes the diagonal part of the kernel of $\gamma$. If the assumption $\gamma\le\id$ is omitted, the constant $K_d$ has to be replaced by $\norm[\infty]{\gamma}^{-2/d}K_d$ where $\norm[\infty]{\gamma}$ is the largest eigenvalue of $\gamma$.
\end{remark}

While the duality discussed above provided Lieb and Thirring with a practical way of proving the kinetic energy inequality \eqref{eq:kineticas} via \eqref{eq:LT}, direct proofs of the bound were discovered later. Eden and Foias \cite{Eden1991} provided a proof in dimension $d=1$, which was generalised to a matrix-valued setting and subsequently lifted to higher dimensions by Dolbeault, Laptev and Loss \cite{Dolbeault2008}. A proof for all $d\ge1$ by Rumin \cite{Rumin2011} was recently improved by Frank, Hundertmark, Jex and Nam \cite{Frank2021b}. Other direct proofs were given by Lundholm and Solovej \cite{Lundholm2013} as well as Sabin \cite{Sabin2016}. We will give the main ideas of the arguments in \cite{Frank2021b}. It is convenient to prove the inequality in the version of Remark \ref{rem:kineticg}. 
\begin{proof}[Direct proof of the kinetic energy inequality]
With any function $f:\R_+\to\R_+$ satisfying $\int_{\R_+}f^2\dd s=1$ one can write 
\begin{align*}
-\Delta=\int_0^\infty f(s/p^2)^2\dd s
\end{align*} 
where $p=-\ii\nabla$. By the Cauchy--Schwarz inequality and $0\le\gamma\le \id$ the operator inequality
\begin{align*}
\gamma
\le(1+\eps)f(s/p^2)\gamma f(s/p^2)+(1+1/\eps)(1-f(s/p^2))^2
\end{align*}
holds for all $\eps>0$. In particular, the diagonal $\gamma(x,x)$ of the kernel of $\gamma$ is bounded by the corresponding diagonal of the kernel on the right side. Optimising over $\eps$ (which may now depend on $x$), one obtains
\begin{align*}
\sqrt{\gamma(x,x)}&\le\sqrt{(f(s/p^2)\gamma f(s/p^2))(x,x)}+\sqrt{(1-f(s/p^2))^2(x,x)}
\end{align*}
The third square root can be computed explicitly to be of the form $c_fs^{d/4}$ with $c_f>0$. Putting everything together yields a bound of the desired form
\begin{align*}
\tr(-\Delta\gamma)&=\int_{\R^d}\int_0^\infty (f(s/p^2)\gamma f(s/p^2))(x,x)\dd s\dd x\\
&\ge \int_{\R^d}\int_0^\infty (\sqrt{\gamma(x,x)}-c_f s^{d/4})_+^2\dd s\dd x
=c_f^{-4/d}\frac{d^2}{(d+2)(d+4)}\int_{\R^d}\gamma(x,x)^{1+2/d}\dd x\,.
\end{align*}
The choice of $f$ can still be optimised over.
\end{proof}

The advance of direct proofs of the kinetic energy inequality has made it possible to establish semiclassical inequalities in new settings. For example, in \cite{Frank2011,Frank2013} Frank, Lewin, Lieb and Seiringer proved a kinetic energy inequality in the presence of a constant background density which by duality yields a generalisation of \eqref{eq:LT} for potentials converging to a positive constant at infinity. It also opened up a new research area of studying kinetic inequalities where the dual Lieb--Thirring inequalities do not even exist. This was initiated by the case of anyons in \cite{Lundholm2013}. We refer to Nam's review \cite{Nam2020} and references therein for details of these developments.
 
As we will discuss in Section \ref{subsec:state}, several of the currently best bounds on the constant $L_{\gamma,d}$ in \eqref{eq:LT} have been obtained through direct proofs of the kinetic energy inequality. 


\section{The Cwikel--Lieb--Rozenblum bound}\label{sec:CLR}
For $\gamma=0$ inequality \eqref{eq:LT} provides an upper bound on the number of negative eigenvalues of a Schr{\"o}dinger operator $-\Delta+V$. It was proved independently (and with different arguments) by Cwikel \cite{Cwikel1977}, Lieb \cite{Lieb1976b,Lieb1980} and Rozenblum \cite{Rozenblum1972} and is thus referred to as the CLR bound.
\begin{theorem}[The CLR bound]
Let $d\ge3$. There is a constant $L_{0,d}$ such that for all real-valued $V\in L^{d/2}(\R^d)$ the number $N_0$ of negative eigenvalues of $-\Delta+V$ satisfies
\begin{align*}
N_0\le L_{0,d}\int_{\R^d} V(x)_-^{d/2}\dd x\,.
\end{align*}
\end{theorem}

Over the last 45 years new proofs of the CLR bound have been found by Fefferman \cite{Fefferman1983}, Conlon \cite{Conlon1985}, Li and Yau \cite{Li1983}, and Frank \cite{Frank2014}.  
However, Lieb's original argument \cite{Lieb1976b} still provides the smallest constants $L_{0,d}$ in dimensions $d=3,4$.  In dimensions $d\ge5$ Lieb's constants have only very recently been improved by Hundertmark, Kunstmann, Ried and Vugalter \cite{Hundertmark2018} using a refinement of Cwikel's approach.  A table of the currently best constants up to $d=9$ can be found in \cite[Table 1]{Hundertmark2018} (partly based on \cite[Table 3.1]{Roepstorff1994}) from where we cite the following upper bounds, presented in the notation of Remark \ref{rem:sc},
\begin{align*}
L_{0,3}\le 6.86924 L_{0,3}^{\mathrm{cl}}\text{ i.e. } R_{0,3}&\le 6.86924,\\
L_{0,4}\le 6.03398 L_{0,4}^{\mathrm{cl}}\text{ i.e. } R_{0,4}&\le 6.03398,\\
L_{0,5}\le 5.95405 L_{0,5}^{\mathrm{cl}}\text{ i.e. } R_{0,5}&\le 5.95405,\\
L_{0,6}\le 5.77058 L_{0,6}^{\mathrm{cl}}\text{ i.e. } R_{0,6}&\le 5.77058.
\end{align*}
For higher dimensions $d\ge7$ we note that $R_{0,d}\le5.77058$ as per \cite[Theorem 1.7]{Hundertmark2018}.

To obtain a lower bound on $L_{0,d}$ one can again consider the analogous one-particle problem i.e.~the problem of finding a constant $L_{0,d}^{\mathrm{one}}$ such that
\begin{align*}
 N_0=0 \quad \text{if} \quad L_{0,d}^{\mathrm{one}}\int_{\R^d} V(x)_-^{d/2}\dd x<1\,.
\end{align*}
Clearly $L_{0,d}\ge L_{0,d}^{\mathrm{one}}$. Similarly as in the case of $\gamma=1$, one can prove that the smallest $L_{0,d}^{\mathrm{one}}$ is related to the optimal constant in an uncertainty principle, more precisely $L_{0,d}^{\mathrm{one}}=S_d^{-d/2}$ in the Sobolev inequality
\begin{align}
\int_{\R^d}|(\nabla \psi)(x)|^2\dd x\ge S_d\left(\int_{\R^d}|\psi(x)|^{\frac{2d}{d-2}}\dd x\right)^{\frac{d-2}{d}}
\label{eq:Sobolev}
\end{align}
for functions $\psi\in L^2(\R^d)$ with $\nabla\psi\in L^2(\R^d)$. 
The optimal values of $S_d$ are known explicitly in all dimensions. In particular $L_{0,3}^{\mathrm{one}}=8L_{0,3}^{\mathrm{cl}}/\sqrt{3}$ and thus \begin{align*}
\frac{8}{\sqrt{3}} L_{0,3}^{\mathrm{cl}}\le L_{0,3}\le 6.86924L_{0,3}^{\mathrm{cl}}\text{ i.e. }
\frac{8}{\sqrt{3}} \le R_{0,3}\le 6.86924\,.
\end{align*}
Note that $8/\sqrt{3}\approx4.6188$. As we will discuss below, the optimal $L_{0,3}$ is conjectured to coincide with this lower bound.

While the CLR bound implies the Sobolev inequality, under certain assumptions in an abstract setting the converse is also true, as observed by Levin and Solomyak \cite{Levin1997}. 
More generally, Frank, Lieb and Seiringer \cite{Frank2010} showed a similar equivalence between Lieb--Thirring inequalities \eqref{eq:LT} and Gagliardo--Nirenberg inequalities. The authors used this to establish Lieb--Thirring inequalities for fractional Laplacians with Hardy term \cite{Frank2008}. A simpler proof of their results without employing this equivalence was later given by Frank \cite{Frank2009}. We remark that there is by now a large literature on Hardy--Lieb--Thirring inequalities for (fractional) Laplacians where a positive Hardy weight is subtracted from the operator, initiated by Ekholm and Frank \cite{Ekholm2006}. For a review, we refer to \cite{Frank2020b}. 


\section{The Lieb--Thirring constants $L_{\gamma,d}$}
\subsection{The conjecture}
Comparing the Lieb--Thirring inequality \eqref{eq:LT} to the Weyl asymptotics \eqref{eq:Weyl} one immediately obtains the lower bound
\begin{align*}
L_{\gamma,d}\ge L_{\gamma,d}^{\mathrm{cl}}\,.
\end{align*}
Equivalently, in the notation of Remark \ref{rem:sc}, $R_{\gamma,d}\ge 1$. 
As already observed above in the cases $\gamma=0$ and $\gamma=2$, a different lower bound can be obtained by investigating the best constant $L_{\gamma,d}^{\mathrm{one}}$ in the corresponding one-particle bound 
\begin{align*}
|E_1|^\gamma\le L_{\gamma,d}^{\mathrm{one}}\int_{\R^d} V(x)_-^{\gamma+d/2}\dd x
\end{align*}
concerning only the lowest eigenvalue. Clearly
\begin{align*}
L_{\gamma,d}\ge L_{\gamma,d}^{\mathrm{one}}\,.
\end{align*}
The variational problem for $L_{\gamma,d}^{\mathrm{one}}$ was first considered by Keller \cite{Keller1961}. Independently Lieb and Thirring \cite{Lieb1976}  showed the optimal constants $L_{\gamma,d}^{\mathrm{one}}$ are related to optimal constants in Gagliardo--Nirenberg inequalities (generalisations of \eqref{eq:GN} and \eqref{eq:Sobolev}).  
In one dimension $d=1$ they solved the variational problem explicitly. The optimal constant is
\begin{align*}
L_{\gamma,1}^{\mathrm{one}}=2\left(\frac{\gamma-1/2}{\gamma+1/2}\right)^{\gamma-1/2}L_{\gamma,1}^{\mathrm{cl}}
\end{align*}
which is, if $\gamma>1/2$, achieved for the P{\"o}schl--Teller potential \cite{Poschl1933}
\begin{align}
V_\gamma(x)=-\frac{(\gamma-1/2)(\gamma+1/2)}{\cosh(x)^2}\,.
\label{eq:cosh}
\end{align}
In their original paper \cite{Lieb1976}, Lieb and Thirring conjectured the following.

\begin{conjecture}[Lieb--Thirring conjecture]
The optimal $L_{\gamma,d}$ in \eqref{eq:LT} is given by
\begin{align*}
L_{\gamma,d}=\max(L_{\gamma,d}^{\mathrm{cl}},L_{\gamma,d}^{\mathrm{one}})\,.
\end{align*}
\end{conjecture}
\noindent In the notation of Remark \ref{rem:sc} the conjecture claims that $R_{\gamma,d}=\max(1,L_{\gamma,d}^{\mathrm{one}}/L_{\gamma,d}^{\mathrm{cl}})$. 
While the conjecture has been proved correct in some cases, it also known to fail in others. Before we review the results, we will remark on two special cases of the conjecture, which are considered especially important.

In dimension $d=1$ the explicit form of $L_{\gamma,1}^{\mathrm{one}}$ allows one to conclude that $\max(L_{\gamma,1}^{\mathrm{cl}},L_{\gamma,1}^{\mathrm{one}})$ coincides with $L_{\gamma,1}^{\mathrm{one}}$ for $1/2\le\gamma\le3/2$ and with $L_{\gamma,1}^{\mathrm{cl}}$ for $\gamma\ge3/2$. The Lieb--Thirring conjecture in one dimension can thus be rephrased as follows.
\begin{conjecture}[Lieb--Thirring conjecture: $d=1$] The optimal $L_{\gamma,1}$ in \eqref{eq:LT} is given by
\begin{align*}
L_{\gamma,1}=\begin{cases}
2\left(\frac{\gamma-1/2}{\gamma+1/2}\right)^{\gamma-1/2}L_{\gamma,1}^{\mathrm{cl}}\,,&1/2\le\gamma\le3/2\\
L_{\gamma,1}^{\mathrm{cl}}\,,&3/2\le\gamma\,.
\end{cases}
\end{align*}
\end{conjecture}
\noindent The conjecture has been proved for $\gamma=1/2$ and $\gamma\ge3/2$ but remains open in the parameter range $1/2<\gamma<3/2$. 

In dimensions $d\le7$ Frank, Gontier and Lewin \cite{Frank2020} proved that similarly $\max(L_{\gamma,d}^{\mathrm{cl}},L_{\gamma,d}^{\mathrm{one}})$ coincides with $L_{\gamma,d}^{\mathrm{one}}$ for $\gamma\le\gamma_{d}^{\mathrm{c}}$ and with $L_{\gamma,1}^{\mathrm{cl}}$ for $\gamma\ge\gamma_{d}^{\mathrm{c}}$ for some critical $\gamma_{d}^{\mathrm{c}}$.   If $d\ge 8$ then the maximum is always attained at $L_{\gamma,d}^{\mathrm{cl}}$. Numerical values of $\gamma_{d}^{\mathrm{c}}$ can be found in \cite{Frank2020} and are included in Figure \ref{fig:LTconjecture} below. In particular $\gamma_{d}^{\mathrm{c}}>1$ for $d\le2$ and $\gamma_{d}^{\mathrm{c}}<1$ for $3\le d\le 7$. From $L_{1,d}^{\mathrm{cl}}$ and $L_{1,d}^{\mathrm{one}}$ we can define $K_{d}^{\mathrm{cl}}$ and  $K_{d}^{\mathrm{one}}$ via \eqref{eq:LK}. By duality the latter can also be defined as the optimal constant in \eqref{eq:GN}. The results of Section \ref{sec:kinetic} imply that the Lieb--Thirring conjecture for $\gamma=1$ can then be phrased as follows.

\begin{conjecture}[Lieb--Thirring conjecture: dual version for $\gamma=1$] The optimal constant in \eqref{eq:kinetic} (and hence also in \eqref{eq:kineticas} and \eqref{eq:kineticg}) is given by
\begin{align*}
K_{d}=
\begin{cases}
K_{d}^{\mathrm{one}}\,,&d\le2\\
K_{d}^{\mathrm{cl}}\,,&d\ge3\,.\\
\end{cases}
\end{align*}
\end{conjecture}
\noindent The conjecture is still open in all dimensions. Of particular interest is the case $d=3$, as $K_{3}=K_{3}^{\mathrm{cl}}$ would imply that the right side in \eqref{eq:kineticas} is precisely the kinetic energy ansatz in Thomas--Fermi theory i.e.~that the quantum mechanical kinetic energy is always bounded from below by the former energy.


\subsection{Monotonicity properties of $L_{\gamma,d}$}
Two monotonicity properties of the optimal constants $L_{\gamma,d}$ have proved very useful in the investigation of the optimal values.  Firstly, the quotients $L_{\gamma,d}/L_{\gamma,d}^{\mathrm{cl}}=R_{\gamma,d}$ are non-increasing with respect to the parameter $\gamma$ as proved by Aizenman and Lieb \cite{Aizenman1978}. 
\begin{theorem}[Monotonicity in $\gamma$: The Aizenman--Lieb principle]\label{th:AL}
Let $d\ge1$ and $\gamma\ge0$ satisfy the assumptions of Theorem \ref{th:LT}. Then, for any $\gamma'\ge\gamma$
\begin{align*}
\frac{L_{\gamma',d}}{L_{\gamma',d}^{\mathrm{cl}}}
\le \frac{L_{\gamma,d}}{L_{\gamma,d}^{\mathrm{cl}}}
\quad\text{i.e.}\quad
R_{\gamma',d}\le R_{\gamma,d}\,.
\end{align*}
\end{theorem}
 As a consequence of the principle, if one can prove that for some $\gamma$ the optimal constant is 
 $R_{\gamma,d}=1$, then in fact $R_{\gamma',d}=1$ for all $\gamma'\ge\gamma$. 
 
\begin{proof}[Proof of Theorem \ref{th:AL}]
The principle is a consequence of the fact that for $\gamma<\gamma'$ and $E\le0$
\begin{align*}
|E|^{\gamma'}=C_{\gamma',\gamma}\int_0^\infty (E+s)_-^{\gamma} s^{\gamma'-\gamma-1}\dd s
\end{align*}
with some constant $C_{\gamma',\gamma}>0$. This can be verified by using the scaling properties of the quantities involved. As a consequence
\begin{align*}
\sum_{j\ge1}|E_j|^{\gamma'} =C_{\gamma',\gamma}\int_0^\infty\sum_{j\ge 1}(E_j+s)_-^{\gamma} s^{\gamma'-\gamma-1}\dd s\,.
\end{align*}
By the variational principle the number $(E_j+s)_-$ is smaller than the modulus of the $j$-th negative eigenvalue of the Schr{\"o}dinger operator with potential $-(V+s)_-$.  Applying the Lieb--Thirring inequality with this potential yields, in the notation of Remark \ref{rem:sc}, the desired 
\begin{align*}
\sum_{j\ge1}|E_j|^{\gamma'} 
&\le C_{\gamma',\gamma}\frac{R_{\gamma,d}}{(2\pi)^d}
\int_0^\infty\int_{\R^d}\int_{\R^d}(p^2+V(x)+s)_-^{\gamma} \dd x \dd p s^{\gamma'-\gamma-1}\dd s\\
&=\frac{R_{\gamma,d}}{(2\pi)^d} \int_{\R^d}\int_{\R^d}(p^2+V(x))_-^{\gamma'} \dd x \dd p\,.
\qedhere
\end{align*}
\end{proof}

A monotonicity type result also holds with respect to the dimension $d$.  The underlying argument was first introduced by Laptev \cite{Laptev1997} and then applied to Lieb--Thirring inequalities by Laptev and Weidl \cite{Laptev2000}. The main idea of this so-called Laptev--Weidl lifting argument is to ``split off'' dimensions of the Laplacian and include them in the potential. To this end the inequality has to be generalised to operator-valued potentials. With $\Hs $ denoting a separable Hilbert space, one can consider the operator $-\Delta\otimes\id_\Hs+V$ on $L^2(\R^d;\Hs )$. Here, $V$ is a function on $\R^d$ that takes values in the set of compact, self-adjoint operators on $\Hs$. One can then show (see \cite{Hundertmark2002,Hundertmark2000} for the endpoint cases) that for $d$ and $\gamma$ as in Theorem \ref{th:LT} there exists a constant $L_{\gamma,d}^{\mathrm{op}}$ (independent of $\Hs$) such that the negative eigenvalues $E_j$ of $-\Delta\otimes\id_\Hs+V$ satisfy
\begin{align}
\sum_{j\ge1} |E_j|^\gamma\le L_{\gamma,d}^{\mathrm{op}}\int_{\R^d}\tr_{\Hs } V(x)_-^{\gamma+d/2}\dd x\,,
\label{eq:LTop}
\end{align} 
provided the integral is finite. Equivalently, with $R_{\gamma,d}^{\mathrm{op}}=L_{\gamma,d}^{\mathrm{op}}/L_{\gamma,d}$,
\begin{align}
\tr_{L^2(\R^d;\Hs)}(-\Delta\otimes\id_\Hs+V)_-^\gamma\le \frac{R_{\gamma,d}^{\mathrm{op}}}{(2\pi)^d}\int_{\R^d}\int_{\R^d}\tr_{\Hs}(p^2+V(x))_-^\gamma\dd p\dd x\,.
\label{eq:LTopR}
\end{align}
Since the choice $\Hs=\C$ corresponds to the previously discussed scalar case, the bounds $L_{\gamma,d}\le L_{\gamma,d}^{\mathrm{op}}$ hold for all admissible $d$ and $\gamma$. Below we recall the proof of the Laptev--Weidl lifting argument in a version given by Hundertmark \cite{Hundertmark2002} that considers ``splitting off'' several dimensions at a time.

\begin{theorem}[Monotonicity in $d$: The Laptev--Weidl lifting argument]\label{th:LW}
Let $d\ge1$ and $\gamma\ge0$ satisfy the assumptions of Theorem \ref{th:LT}. Then, for $d'> d$
\begin{align*}
\frac{L_{\gamma,d'}^{\mathrm{op}}}{L_{\gamma,d'}^{\mathrm{cl}}}
\le \frac{L_{\gamma+d/2,d'-d}^{\mathrm{op}}}{L_{\gamma+d/2,d'-d}^{\mathrm{cl}}}
\frac{L_{\gamma,d}^{\mathrm{op}}}{L_{\gamma,d}^{\mathrm{cl}}}
\quad\text{i.e.}\quad
R_{\gamma,d'}^{\mathrm{op}}\le R_{\gamma+d/2,d'-d}^{\mathrm{op}}R_{\gamma,d}^{\mathrm{op}}\,,
\end{align*}
and in particular, if $d'-d$ and $\gamma$ also satisfy the assumptions of Theorem \ref{th:LT}, then
\begin{align*}
\frac{L_{\gamma,d'}^{\mathrm{op}}}{L_{\gamma,d'}^{\mathrm{cl}}}
\le \frac{L_{\gamma,d'-d}^{\mathrm{op}}}{L_{\gamma,d'-d}^{\mathrm{cl}}}
\frac{L_{\gamma,d}^{\mathrm{op}}}{L_{\gamma,d}^{\mathrm{cl}}}
\quad\text{i.e.}\quad
R_{\gamma,d'}^{\mathrm{op}}\le R_{\gamma,d'-d}^{\mathrm{op}}R_{\gamma,d}^{\mathrm{op}}\,.
\end{align*}
\end{theorem}
 As a consequence of the lifting argument, if one can prove that $R_{\gamma,1}^{\mathrm{op}}=1$ for some $\gamma\ge1/2$ then also
 \begin{align*}
1\le R_{\gamma,d}^{\mathrm{op}}\le R_{\gamma,d-1}^{\mathrm{op}} R_{\gamma,1}^{\mathrm{op}}= R_{\gamma,d-1}^{\mathrm{op}}\,.
\end{align*}
Thus, by induction, $R_{\gamma,d}^{\mathrm{op}}=1$ and hence $R_{\gamma,d}=1$ for all $d\ge 1$.

\begin{proof}[Proof of Theorem \ref{th:LW}]
With $d_1=d$ and $d_2=d'-d$ one decomposes $x=(x_1,x_2)$ and correspondingly  $\Delta=\Delta_{1}+\Delta_{2}$. The space $L^2(\R^{d'})$ is isomorphic to $L^2(\R^{d_1};\Hs)$ with $\Hs =L^2(\R^{d_2})$ and thus one can identify $-\Delta+V$ with $-\Delta_{1}\otimes\id+ W$ where $W$ is the operator-valued potential 
\begin{align*}
W(x_1)=-\Delta_{2}+V(x_1,\cdot)\,.
\end{align*}
Since $\tr_{L^2(\R^{d'})}(-\Delta+V)_-^\gamma
=\tr_{L^2(\R^{d_1};\Hs )}(-\Delta_{1}+W)_-^\gamma$ the Lieb--Thirring inequality \eqref{eq:LTop} yields
\begin{align*}
\tr_{L^2(\R^{d'})}(-\Delta+V)_-^\gamma
\le L_{\gamma,d_1}^{\mathrm{op}}\int_{\R^{d_1}}\tr_{\Hs}W(x_1)_-^{\gamma+d_1/2}\dd x_1\,.
\end{align*}
Noting that $\tr_{\Hs}W(x_1)_-^{\gamma+d_1/2}=\tr_{L^2(\R^{d_2})}(-\Delta_2+V(x_1,\cdot))_-^{\gamma+d_1/2}$
one can also apply \eqref{eq:LTop} to the integrand 
\begin{align*}
\tr_{\Hs}W(x_1)_-^{\gamma+d_1/2}\le L_{\gamma+d_1/2,d_2}^{\mathrm{op}}\int_{\R^{d_2}}V(x_1,x_2)_-^{\gamma+d_1/2+d_2/2}\dd x_2
\end{align*}
and thus $L_{\gamma,d'}^{\mathrm{op}}\le L_{\gamma,d_1}^{\mathrm{op}}L_{\gamma+d_1/2,d_2}^{\mathrm{op}}$.  By direct computation $L_{\gamma,d'}^{\mathrm{cl}}= L_{\gamma,d_1}^{\mathrm{cl}}L_{\gamma+d_1/2,d_2}^{\mathrm{cl}}$, which proves the first result of the theorem.

The second result can then either be obtained from the first by noting that the Aizenman--Lieb principle extends to the operator-valued setting or be proved directly by a similar argument as above using \eqref{eq:LTopR}.
\end{proof}
 

\subsection{Proof of the conjecture for $d\ge 1$ and $\gamma\ge 3/2$}
The Lieb--Thirring conjecture is true in any dimension $d\ge 1$ for $\gamma\ge 3/2$. The proof showcases the usefulness of the monotonicity properties above. 

The case $d=1$ and $\gamma=3/2$ predates Lieb and Thirring's original paper and was already established by Gardner, Greene, Kruskal and Miura \cite{Gardner1961}. They relied on the Buslaev--Faddeev--Zaharov trace formula \cite{Buslaev1960,Zaharov1971} 
\begin{align*}
\sum_{j\ge1} |E_j|^{3/2}+\frac{3}{\pi}\int_{\R_+}k^2 \log |a(k)|\dd k=\frac{3}{16}\int_{\R} V(x)^2\dd x\,,
\end{align*}
for the negative eigenvalues $E_j$ of the Schr{\"o}dinger operator $-\frac{\mathrm{d}^2}{\mathrm{d}x^2}+V$. Noting that the scattering coefficient $a(k)$ (which coincides with the perturbation determinant) satisfies $|a(k)|\ge 1$, the authors of \cite{Gardner1961} obtained a Lieb--Thirring inequality with the constant $L_{3/2,1}=3/16$. This constant coincides with $L_{3/2,1}^{\mathrm{cl}}$, proving that the bound is sharp and that the Lieb--Thirring conjecture holds for these parameters. 
Lieb and Thirring \cite{Lieb1976} recovered this result and showed that the inequality becomes an equality for the potential \eqref{eq:cosh}. They also extended $L_{\gamma,1}=L_{\gamma,1}^{\mathrm{cl}}$ to all $\gamma=3/2+k$ with $k\in\N$. Using their principle (Theorem \ref{th:AL}), Aizenman and Lieb \cite{Aizenman1978} subsequently extended the result to all $\gamma\ge3/2$ proving the conjecture in one dimension for this parameter range.

Laptev and Weidl \cite{Laptev2000} later considered the operator-valued setting described in the previous subsection and generalised the above trace formula to the case $\Hs=\C^n$, i.e.~to the Hamiltonian $-\frac{\mathrm{d}^2}{\mathrm{d}x^2}\otimes\id_{\C^n}+V$ on $L^2(\R;\C^n)$. Using again the positivity of the term involving the scattering coefficient, they obtained the sharp inequality
\begin{align}
\sum_{j\ge 1} |E_j|^{3/2}\le\frac{3}{16}\int_{\R}\tr_{\Hs}V(x)^2\dd x\,
\label{eq:LTLW}
\end{align}
first for the case $\Hs=\C^n$ and subsequently for general separable Hilbert spaces $\Hs$. Thus $R_{3/2,1}^{\mathrm{op}}=1$ and from their lifting argument (Theorem \ref{th:LW}) Laptev and Weidl could then conclude that $R_{3/2,d}=1$ for all $d\ge1$. By the Aizenman--Lieb principle the conjecture $R_{\gamma,d}=1$ thus holds for all $\gamma\ge3/2$.

The trace formula mentioned above and its extensions to higher powers of the eigenvalues were historically of great importance in establishing integrability of the Korteweg--De Vries equation. They have also proved very useful in questions relating to the absolute continuity of the spectrum of Schr{\"o}dinger operators \cite{Deift1999,Killip2009}.   
An elegant direct proof of \eqref{eq:LTLW} without first establishing the trace formula was given by Benguria and Loss \cite{Benguria2000} soon after the work of Laptev and Weidl. The argument relies on the so-called commutation method to successively remove eigenvalues from the spectrum of the operator. This method has also been applied to establish sharp Lieb--Thirring type inequalities for Schr{\"o}dinger operators on the half-axis with Robin boundary condition \cite{Exner2014}, for  fourth order operators \cite{Hoppe2006} and for discrete Jacobi operators \cite{Schimmer2015}. For the latter operators, corresponding trace formulae were established by Deift and Killip \cite{Deift1999} as well as Killip and Simon \cite{Killip2003}. Using an approximation argument \cite{Schimmer2015}, the Lieb--Thirring bound in the discrete setting can be used to provide an alternative proof of $R_{3/2,1}=1$.


\subsection{Proof of the conjecture for $d=1$ and $\gamma=1/2$}\label{sec:WHLT}
The Lieb--Thirring inequality \eqref{eq:LT} in the endpoint case $d=1,\gamma=1/2$ was first proved by Weidl \cite{Weidl1996},  20 years after the original paper by Lieb and Thirring. His argument showed that $L_{1/2,1}\le 1.005$. Subsequently Hundertmark, Lieb and Thomas \cite{Hundertmark1998} provided a proof that included the sharp constant $L_{1/2,1}=L_{1/2,1}^{\mathrm{one}}=1/2$, i.e. $R_{1/2,1}=2$ thus proving the conjecture for these parameters. To date, this is the only case where the sharp value of $L_{\gamma,d}>L_{\gamma,d}^{\mathrm{cl}}$ is known. We recall their proof below. 

\begin{theorem}\label{th:LT12}
For all real-valued $V\in L^{1}(\R)$ the negative eigenvalues $E_j$ of the operator $-\DD{x}+V$ satisfy
\begin{align*}
\sum_{j\ge 1}|E_j|^{1/2}\le\frac12\int_{\R} V(x)_-\dd x\,.
\end{align*}
\end{theorem}

\begin{proof}
The kernel $G_E(x-y)$ of the resolvent $(-\frac{\mathrm{d}^2}{\mathrm{d}x^2}-E)^{-1}$ has a singularity $G_{E}(0)=1/(2\sqrt{|E|})$ on the diagonal that leads to a diverging integral in the original proof of Theorem \ref{th:LT}. This motivates to consider a modified Birman--Schwinger operator 
\begin{align*}
L_E=G_E(0)^{-1}K_{E}\,.
\end{align*}
From the corresponding properties of $K_E$ one immediately obtains
\begin{enumerate}[(i)]
\item $L_E\ge0$ is compact with eigenvalues $\mu_1(L_E)\ge\mu_2(L_E)\ge\dots\ge0$,
\item $\mu_j(L_{E_j})=G_{E_j}(0)^{-1}=2|E_j|^{1/2}$ \emph{(Birman--Schwinger principle)}.
\end{enumerate}
While monotonicity $L_{E}< L_{E'}$ is not possible (since both operators have the same trace), Lieb, Hundertmark and Thomas observed that it is sufficient to replace the monotonicity (iii) of individual eigenvalues by monotonicity of  partial eigenvalue sums $S_N(L_E)=\sum_{j\ge1}^N \mu_j(L_{E})$, i.e.~by
\begin{enumerate}[(i)$'$] \setcounter{enumi}{2}
\item$S_N(L_{E})\le S_N(L_{E'})$ for $E\le E'\le0$ and $N\in\N$. 
\end{enumerate}
Assuming this property to be true, one observes that
\begin{align*}
G_{E_1}(0)^{-1}&=\mu_1(L_{E_1})\le\mu_1(L_{E_2})\,,\\
\sum_{j=1}^2G_{E_j}(0)^{-1}&=\sum_{j=1}^2\mu_j(L_{E_j})\le\sum_{j=1}^2\mu_j(L_{E_2})
\le\sum_{j=1}^2\mu_j(L_{E_3})\,. 
\end{align*}
Continuing in this manner, one obtains for any $n\ge1$
\begin{align*}
\sum_{j=1}^nG_{E_j}(0)^{-1}\le\sum_{j=1}^n\mu_j(L_{E_n})\le\tr_{L^2(\R)} L_{E_n}=\int_{\R} |V(x)|\dd x
\end{align*}
which implies the desired result. The monotonicity (iii)$'$ was proved in \cite{Hundertmark1998} by means of an approximation argument and careful analysis of the dependence of the kernel on the spectral parameter $E$. An alternative proof was given later by Hundertmark, Laptev and Weidl \cite{Hundertmark2000} who noted that
\begin{align}
L_E=\int_{\R}U(p)^*L_0U(p) \widehat{g}_E(p)\dd p
\label{eq:LE}
\end{align}
where $L_0$ is the rank-one operator with kernel $\sqrt{|V(x)|}\sqrt{|V(y)|}$, $U(p)$ is the multiplication operator $\e^{-2\pi\ii  x p }$ and $\widehat{g}_E(p)=G_E(0)^{-1}((2\pi p)^2-E)^{-1}$ is the Fourier transform of $G_E(0)^{-1}G_E(x)$. 
Using the explicit form $G_{E}(0)^{-1}G_E(x)=\e^{-\sqrt{|E|}x}$ one obtains for $E<E'<0$ the convolution property
\begin{align}
\widehat{g}_E*\widehat{g}_{E,E'}=\widehat{g}_{E'}\,,
\label{eq:conv}
\end{align}
where $\widehat{g}_{E,E'}=\widehat{g}_{-(\sqrt{|E|}-\sqrt{|E'|})^2}$ . Importantly all three functions are probability densities, i.e.~they are non-negative and integrate to one, with the latter being a consequence of the factor $G_E(0)^{-1}$. Using that $U(p+p')=U(p)U(p')$ one obtains
\begin{align*}
L_E=\int_{\R}U(p)^*L_{E'}U(p) \widehat{g}_{E,E'}(p)\dd p
\end{align*}
from \eqref{eq:LE}. The monotonicity (iii)$'$ then follows from Ky-Fan's inequality, which states that $S_N$ satisfies a triangle inequality.
\end{proof}

The inequality becomes an equality for delta potentials $c\delta$ with $c<0$. The result was extended to operator-valued potentials by Hundertmark, Laptev and Weidl \cite{Hundertmark2000}. The convolution property \eqref{eq:conv} may be phrased as the well known fact that the Cauchy distribution is a convolution semigroup. Similar properties also hold for the resolvent kernels of discrete Jacobi operators \cite{Laptev2021} as well as certain functional difference operators \cite{Laptev2021b}, which can be used to establish monotonicity of partial eigenvalue sums and thus also sharp Lieb--Thirring type inequalities for these operators. In the former case, the monotonicity and ensuing bound had also been proved earlier by Hundertmark and Simon \cite{Hundertmark2002b}  following the arguments of \cite{Hundertmark1998}. Using this discrete result together with an approximation argument \cite{Schimmer2015} yields an alternative proof  of Theorem \ref{th:LT12}.


\subsection{The state of the conjecture}\label{subsec:state}
Figure \ref{fig:LTconjecture} visualises the current state of the conjecture. The previously discussed cases $d=1,\gamma=1/2$ and $d\ge1,\gamma\ge3/2$ are the only parameter regions where the conjecture is proved and the only regions where the optimal constant is known. 

Helffer and Robert \cite{Helffer1990} constructed examples that show
\begin{align*}
L_{\gamma,d}>L_{\gamma,d}^{\mathrm{cl}} \quad\text{if} \quad d\ge 2\text{ and } 0\le\gamma<1\,.
\end{align*}
Note that this inequality also holds for $d=1$ if $1/2\le\gamma<3/2$ since in this parameter range by explicit computation $L_{\gamma,1}^{\mathrm{one}}>L_{\gamma,1}^{\mathrm{cl}}$.
Recalling the definition of $\gamma_{d}^{\mathrm{c}}$ and noting that $0<\gamma_{d}^{\mathrm{c}}<1$ for $3\le d\le 7$,  Helffer and Robert's results establish that the conjecture fails in dimensions $3\le d\le 7$ if $\gamma_{d}^{\mathrm{c}}\le\gamma<1$ as well as in dimensions $d\ge 8$ if $0\le\gamma<1$. Glaser, Grosse and Martin \cite{Glaser1978} proved that 
\begin{align*}
L_{0,d}>L_{0,d}^{\mathrm{one}}\quad \text{if}\quad d\ge 7,
\end{align*}
which shows that the conjecture fails in dimensions $d\ge7$ if $\gamma=0$. More recently Frank, Gontier and Lewin \cite{Frank2020} showed that also
\begin{align*}
L_{\gamma,d}>L_{\gamma,d}^{\mathrm{one}}\quad \text{if}\quad 
d\ge 1\text{ and } \gamma>\max(0,2-d/2)
\end{align*}
disproving the conjecture in dimension $d=2$ if $1<\gamma\le\gamma_{2}^{\mathrm{c}}$, in dimension $d=3$ if $1/2<\gamma<\gamma_{3}^{\mathrm{c}}$ and in dimensions $4\le d\le 7$ if $0<\gamma<\gamma_{d}^{\mathrm{c}}$. 
Numerical discussions of the validity of the conjecture can be found in an appendix to Lieb and Thirring's original paper \cite{Lieb1976} by Barnes, in Levitt's \cite{Levitt2014} and most recently in \cite{Frank2021c} by Frank, Gontier and Lewin. 

We conclude by briefly reviewing the currently best bounds on the Lieb--Thirring constants. The results are again visualised in Figure \ref{fig:LTconjecture}. 
In one dimension $d=1$ for $\gamma<1$ the currently best bounds on $L_{\gamma,1}$ can be obtained by applying the  Aizenman--Lieb principle (Theorem \ref{th:AL}) to the sharp result at $\gamma=1/2$
\begin{align*}
L_{\gamma,1}\le 2L_{\gamma,1}^{\mathrm{cl}}\text{ i.e. }R_{\gamma,1}\le 2\quad\text{if}\quad 1/2\le\gamma<1\,.
\end{align*}
Better bounds are known for $\gamma\ge1$ as Frank, Hundertmark, Jex and Nam \cite{Frank2021b} used their proof method described at the end of Section \ref{sec:kinetic} (together with an additional argument that improves the step where the Cauchy--Schwarz inequality was applied) to establish that $K_{1}\ge 0.471851 K_1^{\mathrm{cl}}$ or equivalently  $L_{1,1}\le 1.456 L_{1,1}^{\mathrm{cl}}$. Thus, by the Aizenman--Lieb principle, 
\begin{align*}
L_{\gamma,1}\le 1.456 L_{\gamma,1}^{\mathrm{cl}}\text{ i.e. }R_{\gamma,1}\le 1.456\quad\text{if}\quad 1\le\gamma<3/2\,.
\end{align*}
The authors of \cite{Frank2021b} also extended their bound to the operator-valued setting, $L_{1,1}^{\mathrm{op}}\le 1.456 L_{1,1}^{\mathrm{cl}}$. As a consequence of the Laptev--Weidl lifting argument (Theorem \ref{th:LW}) and the Aizenman--Lieb principle they then obtained
\begin{align*}
L_{\gamma,d}\le 1.456L_{\gamma,1}^{\mathrm{cl}}\text{ i.e. }R_{\gamma,d}\le 1.456\quad\text{if}\quad d\ge1\text{ and }1\le\gamma<3/2\,,
\end{align*}
as well as
\begin{align*}
L_{\gamma,d}\le 2\cdot1.456L_{\gamma,1}^{\mathrm{cl}}\text{ i.e. }R_{\gamma,d}\le 2\cdot 1.456\quad\text{if}\quad d\ge2\text{ and }1/2\le\gamma<1\,,
\end{align*}
which currently constitute the best bounds in these parameter regions.  
With regards to the conjecture in its dual version for $\gamma=1$ we point out that the above bounds imply that $K_d\ge(0.471851)^{1/d}K_{d}^{\mathrm{cl}}$. In particular $K_3\ge 0.778517 K_3^{\mathrm{cl}}$ as well as $K_1\ge 0.629134 K_1^{\mathrm{one}}$ since $K_1^{\mathrm{cl}}=\pi^2/3$ and $K_1^{\mathrm{one}}=\pi^2/4$. 

The results of \cite{Frank2021b} are only the last in a long line of consecutive improvements of the Lieb--Thirring inequalities. The constants $L_{1,d}$ and $K_d$ have an especially rich history with successive advances made by Lieb and Thirring \cite{Lieb1975}, Lieb \cite{Lieb1980,Lieb1984}, Eden and Foias~\cite{Eden1991}, Blanchard and Stubbe \cite{Blanchard1996}, Hundertmark, Laptev and Weidl \cite{Hundertmark2000} as well as Dolbeault, Laptev and Loss \cite{Dolbeault2008} preceding the results of \cite{Frank2021b}. 

For $d\ge 3$ and $0\le\gamma<1/2$ bounds on $L_{\gamma,d}$ can be obtained by applying the Aizenman--Lieb principle to the best constants in the CLR bound discussed in Section \ref{sec:CLR}. For $d=2$ and $0<\gamma<1/2$ bounds were obtained in \cite{Lieb1976}.

Finally, we remark here that Lieb--Thirring inequalities can be extended to magnetic Schr{\"o}dinger operators, where $-\Delta$ is replaced by $(-\ii\nabla+A)^2$ with a magnetic vector potential $A\in L^2_{\mathrm{loc}}(\R^d;\R^d)$. By the diamagnetic inequality, the proof of Theorem \ref{th:LT} in Section \ref{sec:intro} can immediately be generalised to this case with the same constants.  
It is still unclear however, whether the optimal constant $L_{\gamma,d}$ in \eqref{eq:LT} can be used in a Lieb--Thirring bound with arbitrary $A\in L^2_{\mathrm{loc}}(\R^d;\R^d)$. Since in dimension $d=1$ magnetic fields can be gauged away,  all currently best constants that were achieved by applying the Laptev--Weidl lifting argument to a bound in one dimension remain true in the magnetic case \cite{Laptev2000}. As observed by Avron, Herbst and Simon \cite{Avron1978}, the same is true for Lieb's constants in the CLR bound.

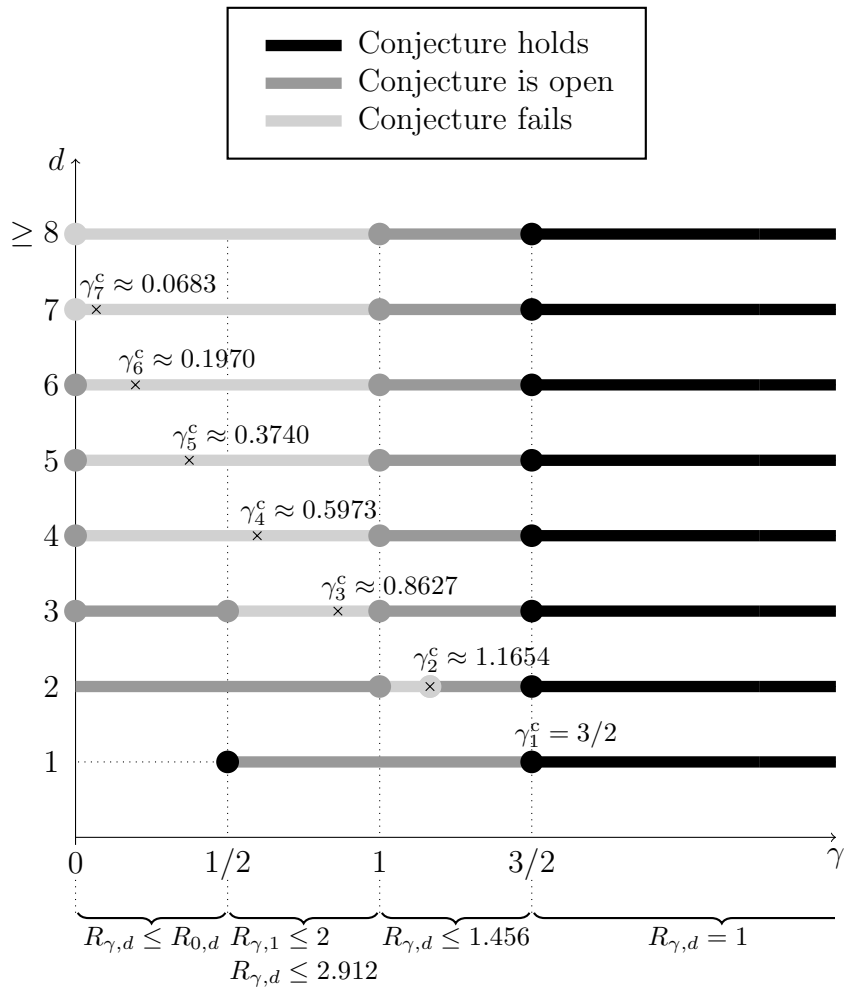
\begin{figure}[!th]
\begin{center}
\begin{tikzpicture}[scale=1]
\draw[->] (0,0)--(10,0) node[below]{$\gamma$};
\draw[->](0,-0.1)--(0,9) node[left]{$d$};

\draw[-,dotted](10,1)--(0,1) node[left]{$1$\,};
\draw[-,dotted](10,2)--(0,2) node[left]{$2$\,};
\draw[-,dotted](10,3)--(0,3) node[left]{$3$\,};
\draw[-,dotted](10,4)--(0,4) node[left]{$4$\,};
\draw[-,dotted](10,5)--(0,5) node[left]{$5$\,};
\draw[-,dotted](10,6)--(0,6) node[left]{$6$\,};
\draw[-,dotted](10,7)--(0,7) node[left]{$7$\,};
\draw[-,dotted](10,8)--(0,8) node[left]{$\ge8$\,};
\node[below] at (0,-0.05) {$0$};

\draw[-,dotted](2,8)--(2,0) node[below]{$1/2$};
\draw[-,dotted](4,8)--(4,0) node[below]{$1\vphantom{\setminus}$};
\draw[-,dotted](6,8)--(6,0) node[below]{$3/2$};

\draw[-,line width=0.15cm,ColFail] (4*1,2)--(4*1.1654,2);
\fill[ColFail] (4*1.1654,2) circle (0.15cm);

\draw[-,line width=0.15cm,ColFail] (4*0.5,3)--(4*1,3);

\draw[-,line width=0.15cm,ColFail] (4*0,4)--(4*1,4);

\draw[-,line width=0.15cm,ColFail] (4*0,5)--(4*1,5);

\draw[-,line width=0.15cm,ColFail] (4*0,6)--(4*1,6);

\draw[-,line width=0.15cm,ColFail] (4*0,7)--(4*1,7);
\fill[ColFail] (4*0,7) circle (0.15cm);

\draw[-,line width=0.15cm,ColFail] (4*0,8)--(4*1,8);
\fill[ColFail] (4*0,8) circle (0.15cm);       

\draw[-,line width=0.15cm,ColOpen] (4*0.5,1)--(4*1.5,1);

\fill[ColOpen] (4*1,2) circle (0.15cm);    
\draw[-,line width=0.15cm,ColOpen] (4*0,2)--(4*1,2);
\draw[-,line width=0.15cm,ColOpen] (4*1.1654,2)--(4*1.5,2);

\fill[ColOpen] (4*0,3) circle (0.15cm);       
\fill[ColOpen] (4*1,3) circle (0.15cm);   
\draw[-,line width=0.15cm,ColOpen] (4*0,3)--(4*0.5,3); 
\fill[ColOpen] (4*0.5,3) circle (0.15cm);   
\draw[-,line width=0.15cm,ColOpen] (4*1,3)--(4*1.5,3); 

\fill[ColOpen] (4*0,4) circle (0.15cm);       
\fill[ColOpen] (4*1,4) circle (0.15cm);   
\draw[-,line width=0.15cm,ColOpen] (4*1,4)--(4*1.5,4); 

\fill[ColOpen] (4*0,5) circle (0.15cm);       
\fill[ColOpen] (4*1,5) circle (0.15cm);   
\draw[-,line width=0.15cm,ColOpen] (4*1,5)--(4*1.5,5); 

\fill[ColOpen] (4*0,6) circle (0.15cm);       
\fill[ColOpen] (4*1,6) circle (0.15cm);   
\draw[-,line width=0.15cm,ColOpen] (4*1,6)--(4*1.5,6); 

\fill[ColOpen] (4*1,7) circle (0.15cm);   
\draw[-,line width=0.15cm,ColOpen] (4*1,7)--(4*1.5,7); 

\fill[ColOpen] (4*1,8) circle (0.15cm);   
\draw[-,line width=0.15cm,ColOpen] (4*1,8)--(4*1.5,8); 

\draw[-,line width=0.15cm,ColTrue] (6,1)--(9,1);
\draw[-,line width=0.15cm,ColTrue] (9,1)--(10,1);

\draw[-,line width=0.15cm,ColTrue] (6,2)--(9,2);
\draw[-,line width=0.15cm,ColTrue] (9,2)--(10,2);

\draw[-,line width=0.15cm,ColTrue] (6,3)--(9,3);
\draw[-,line width=0.15cm,ColTrue] (9,3)--(10,3);

\draw[-,line width=0.15cm,ColTrue] (6,4)--(9,4);
\draw[-,line width=0.15cm,ColTrue] (9,4)--(10,4);

\draw[-,line width=0.15cm,ColTrue] (6,5)--(9,5);
\draw[-,line width=0.15cm,ColTrue] (9,5)--(10,5);

\draw[-,line width=0.15cm,ColTrue] (6,6)--(9,6);
\draw[-,line width=0.15cm,ColTrue] (9,6)--(10,6);

\draw[-,line width=0.15cm,ColTrue] (6,7)--(9,7);
\draw[-,line width=0.15cm,ColTrue] (9,7)--(10,7);

\draw[-,line width=0.15cm,ColTrue] (6,8)--(9,8);
\draw[-,line width=0.15cm,ColTrue] (9,8)--(10,8);

\fill[ColTrue] (2,1) circle (0.15cm);  
\fill[ColTrue] (6,1) circle (0.15cm);  
\fill[ColTrue] (6,2) circle (0.15cm);  
\fill[ColTrue] (6,3) circle (0.15cm);  
\fill[ColTrue] (6,4) circle (0.15cm);  
\fill[ColTrue] (6,5) circle (0.15cm);  
\fill[ColTrue] (6,6) circle (0.15cm);  
\fill[ColTrue] (6,7) circle (0.15cm);  
\fill[ColTrue] (6,8) circle (0.15cm);  

\fill[ColFail] (4*1.1654,2) circle (0.15cm);

\begin{scope}[xshift=2cm,yshift=8cm]
\draw[draw=black, thick] (0,3) rectangle ++(5.5,-2);
\draw[-,line width=0.15cm,ColTrue] (0.5,2.5)--(1.5,2.5) node[right,black]{Conjecture holds};
\draw[-,line width=0.15cm,ColOpen] (0.5,2)--(1.5,2) node[right,black]{Conjecture is open};
\draw[-,line width=0.15cm,ColFail] (0.5,1.5)--(1.5,1.5) node[right,black]{Conjecture fails};
\end{scope}

\tikzset{cross/.style={cross out, draw=black, minimum size=2*(#1-\pgflinewidth), inner sep=0pt, outer sep=0pt},
cross/.default={2pt}}

\draw (4*1.5,1) node[cross]{};
\draw (4*1.1654,2) node[cross]{};
\draw (4*0.8627,3) node[cross]{};
\draw (4*0.5973,4) node[cross]{};
\draw (4*0.3740,5) node[cross]{};
\draw (4*0.1970,6) node[cross]{};
\draw (4*0.0683,7) node[cross]{};

\node[above] at (4*1.5,1+0.05) {\footnotesize$\phantom{=3/2}\gamma_{1}^{\mathrm{c}}=3/2$};
\node[above] at (4*1.1654,2+0.05) {\footnotesize$\phantom{\approx1.1654}\gamma_{2}^{\mathrm{c}}\approx1.1654$};
\node[above] at (4*0.8627,3) {\footnotesize$\phantom{\approx0.8627}\gamma_{3}^{\mathrm{c}}\approx0.8627$};
\node[above] at (4*0.5973,4) {\footnotesize$\phantom{\approx0.5973}\gamma_{4}^{\mathrm{c}}\approx0.5973$};
\node[above] at (4*0.3740,5) {\footnotesize$\phantom{\approx0.3740}\gamma_{5}^{\mathrm{c}}\approx0.3740$};
\node[above] at (4*0.1970,6) {\footnotesize$\phantom{\approx0.1970}\gamma_{6}^{\mathrm{c}}\approx0.1970$};
\node[above] at (4*0.0683,7) {\footnotesize$\phantom{\approx0.0683}\gamma_{7}^{\mathrm{c}}\approx0.0683$};

\draw[-,dotted](4*0,-1)--(4*0,-0.5);
\draw[-,dotted](4*0.5,-1)--(4*0.5,-0.5);
\draw[-,dotted](4*1,-1)--(4*1,-0.5);
\draw[-,dotted](4*1.5,-1)--(4*1.5,-0.5);

\draw[decorate, decoration = {brace,mirror,raise=0pt,amplitude=4pt},  thick] (4*0+0.03,-1)--(4*0.5-0.03,-1) node[below=1pt,pos=0.5,black]{\footnotesize$R_{\gamma,d}\le R_{0,d}$};
\draw[decorate, decoration = {brace,mirror,raise=0pt,amplitude=4pt},  thick] (4*0.5+0.03,-1)--(4*1-0.03,-1) node[below=1pt,pos=0.5,black]{\footnotesize$R_{\gamma,1}\le2\phantom{.912}$};
\node[below] at (4*0.75,-1.5) {\footnotesize$R_{\gamma,d}\le 2.912$};
\draw[decorate, decoration = {brace,mirror,raise=0pt,amplitude=4pt},  thick] (4*1+0.03,-1)--(4*1.5-0.03,-1) node[below=1pt,pos=0.5,black]{\footnotesize$R_{\gamma,d}\le1.456$};
\draw[decorate, decoration = {brace,mirror,raise=0pt,amplitude=4pt},  thick] (4*1.5+0.03,-1)--(11-0.03,-1) node[below=1pt,pos=0.5,black]{\footnotesize$R_{\gamma,d}=1\phantom{.000}$};
\fill[white] (10,-1.1) rectangle (11,-1);

\end{tikzpicture}
\end{center}
\caption{The state of the conjecture as of 2021}
\label{fig:LTconjecture}
\end{figure}

\subsection*{Funding}
The author was supported by VR grant 2017-04736 at the Royal Swedish Academy of Sciences.



\begin{thebibliography}{65}

\bibitem{Aizenman1978}
M.~Aizenman, E.~H.~Lieb, 
\emph{On semiclassical bounds for eigenvalues of {S}chr\"odinger operators}, 
Phys. Lett. A \textbf{66}(6), 427--429 (1978).
%
\bibitem{Avron1978}
J.~Avron, I.~Herbst, B.~Simon, 
\emph{Schr{\"o}dinger operators with magnetic fields. I. General interactions}, 
Duke Math. J. \textbf{45}(4), 847--883 (1978).
%
\bibitem{Benguria2000}
R.~Benguria, M.~Loss, 
\emph{A simple proof of a theorem of Laptev and Weidl}, 
Math. Res. Lett. \textbf{7}, 195--203 (2000).
%
\bibitem{Birman1961}
M.~S.~Birman, 
\emph{The spectrum of singular boundary problems}, 
Mat. Sb. \textbf{55}(2), 125--174 (1961), translated in Amer. Math. Soc. Trans. (2) \textbf{53}, 23--80 (1966).
%
\bibitem{Blanchard1996}
Ph.~Blanchard, J.~Stubbe, 
\emph{Bound states for Schr{\"o}dinger Hamiltonians: Phase Space Methods and Applications}, 
Rev. Math. Phys. \textbf{35}, 504--547 (1996).
%
\bibitem{Buslaev1960}
V.~S.~Buslaev, L.~D.~Faddeev, 
\emph{Formulas for traces for a singular Sturm-Liouville differential operator}, 
Dokl. Akad. Nauk SSSR \textbf{132} 13--16 (1960), English translation in Soviet Math. Dokl. \textbf{1}, 451--454 (1960).
%
\bibitem{Cwikel1977} 
M.~Cwikel, 
\emph{Weak type estimates for singular values and the number of bound states of Schr{\"o}dinger operators}, 
Ann. Math. (2), \textbf{106}(1), 93--100 (1977).
%
\bibitem{Conlon1985}
J.~G.~Conlon, 
\emph{A new proof of the Cwikel--Lieb--Rosenbljum bound}, 
Rocky Mountain J. Math. \textbf{15}. 117-122 (1985)
%
\bibitem{Dolbeault2008}
J.~Dolbeault, A.~Laptev, M.~Loss, 
\emph{Lieb--Thirring inequalities with improved constants}, 
J. Eur. Math. Soc. \textbf{10}, 1121--1126 (2008).
%
\bibitem{Deift1999}
P.~Deift, R.~Killip, 
\emph{On the absolutely continuous spectrum of one-dimensional Schr{\"o}dinger operators with square summable potentials}, 
Comm. Math. Phys. \textbf{203}(2), 341--347 (1999).
%
\bibitem{Eden1991}
A.~Eden, C.~Foias, 
\emph{A simple proof of the generalized Lieb--Thirring inequalities in one-space dimension}, 
J. Math. Anal. Appl. \textbf{162}, 250--254 (1991).
%
\bibitem{Ekholm2006}
T.~Ekholm, R.~L.~Frank, 
\emph{On Lieb--Thirring inequalities for Schr{\"o}dinger operators with virtual level}, 
Comm. Math. Phys. \textbf{264}(3), 725--740 (2006).
%
\bibitem{Exner2014}
P.~Exner, A.~Laptev, M.~Usman ,
\emph{On some sharp spectral inequalities for Schr{\"o}dinger operators on semiaxis}, 
Comm. Math. Phys. \textbf{326}(2), 531--541, 2014.
%
\bibitem{Fefferman1983}
C.~L.~Fefferman, 
\emph{The uncertainty principle}, 
Bull. Amer. Math. Soc. (N.S.), \textbf{9}(2), 129-- 206, (1983).
%
\bibitem{Frank2014}
R.~L.~Frank, 
\emph{Cwikel's theorem and the CLR inequality}, 
J. Spectr. Theory \textbf{4}, 1--21 (2014).
%
\bibitem{Frank2009}
R.~L.~Frank, 
\emph{A simple proof of Hardy--Lieb--Thirring inequalities}, 
Comm. Math. Phys. \textbf{290}(2), 789--800 (2009).
%
\bibitem{Frank2020b}
R.~L.~Frank, 
\emph{The Lieb--Thirring inequalities: Recent results and open problems}, 
preprint, arXiv:2007.09326 (2020).
%
\bibitem{Frank2021c}
R.~L.~Frank, D.~Gontier, M.~Lewin, 
\emph{The periodic Lieb-Thirring inequality}, 
In: Partial Differential Equations, Spectral Theory, and Mathematical Physics. The Ari Laptev Anniversary Volume,  135--154, EMS Publishing House, 2021.
%
\bibitem{Frank2020}
R.~L.~Frank, D.~Gontier, M.~Lewin, 
\emph{The Nonlinear Schr{\"o}dinger Equation for Orthonormal Functions II: Application to Lieb--Thirring Inequalities},  Commun. Math. Phys. \textbf{384}, 1783--1828 (2021). 
%
\bibitem{Frank2021b}
R.~L.~Frank, D.~Hundertmark, M.~Jex, P.~T.~Nam, 
\emph{The Lieb--Thirring inequality revisited}, 
J. Eur. Math. Soc. \textbf{23}, 2583--2600 (2021).
%
\bibitem{Frank2021}
R.~L.~Frank, A.~Laptev, T.~Weidl, 
\emph{Schr{\"o}dinger Operators: Eigenvalues and Lieb--Thirring inequalities},
unpublished, (2021).
%
\bibitem{Frank2011}
R.~L.~Frank, M.~Lewin, E.~H.~Lieb, R.~Seiringer, 
\emph{Energy Cost to Make a Hole in the Fermi Sea}, 
Phys. Rev. Lett. \textbf{106}, 150402 (2011).
%
\bibitem{Frank2013}
R.~L.~Frank, M.~Lewin, E.~H.~Lieb, R.~Seiringer, 
\emph{A positive density analogue of the Lieb--Thirring inequality}, 
Duke Math. J. \textbf{162}, 435--495 (2013).
%
\bibitem{Frank2008}
 R.~L.~Frank, E.~H.~Lieb, R.~Seiringer, 
\emph{Hardy-Lieb-Thirring inequalities for fractional Schr{\"o}dinger operators}, 
J. Amer. Math. Soc. \textbf{21}(4), 925--950 (2008).
%
\bibitem{Frank2010}
R.~L.~Frank, E.~H.~Lieb,~R. Seiringer, 
\emph{Equivalence of Sobolev inequalities and Lieb--Thirring inequalities}, 
In: XVIth International Congress on Mathematical Physics, Proceedings of the ICMP held in Prague, August 3-8, 2009, P. Exner (ed.), 523--535, WorldScientific, Singapore, 2010.
%
\bibitem{Gardner1961}
C.~S.~Gardner, J.~M.~Greene, M.~D.~Kruskal, R.~M.~Miura, 
\emph{Korteweg-deVries equation and generalization. VI. Methods for exact solution}, 
Comm. Pure Appl. Math. \textbf{27}, 97--133 (1974).
 %
\bibitem{Glaser1978} 
V.~Glaser, H.~Grosse, A.~Martin, 
\emph{Bounds on the number of eigenvalues of the Schr{\"o}dinger operator}, 
Comm. Math. Phys. \textbf{59}(2), 197--212 (1978).
%
\bibitem{Helffer1990}
B.~Helffer, D.~Robert, 
\emph{Riesz means of bounded states and semi-classical limit connected with a Lieb-Thirring conjecture. II}, 
Ann. Inst. H. Poincar{\'e} Phys. Th{\'e}or. \textbf{53}(2), 139--147 (1990).
%
\bibitem{Hoppe2006}
J.~Hoppe, A.~Laptev, J.~Ostensson, 
\emph{Solitons and the removal of eigenvalues for fourth-order differential operators}, 
Int. Math. Res. Not.,  Art. ID 85050, 1--14 (2006).
%
\bibitem{Hundertmark2002}
D.~Hundertmark, 
\emph{On the number of bound states for Schr{\"o}dinger operators with operator-valued potentials},  
Ark. Mat. \textbf{40}(1), 73--87 (2002).
%
\bibitem{Hundertmark2007}
D.~Hundertmark, 
\emph{Some bound state problems in quantum mechanics}, 
In: Spectral theory and mathematical physics: a Festschrift in honor of Barry Simon's 60th birthday, 463--496, Proc.
Sympos. Pure Math., \textbf{76}, Part 1, Amer. Math. Soc., Providence, RI, 2007.
%
\bibitem{Hundertmark2002b}
D.~Hundertmark, B.~Simon, 
\emph{Lieb--Thirring inequalities for Jacobi matrices}, 
J. Approx. Theory \textbf{118}(1), 106--130 (2002). 
%
\bibitem{Hundertmark2018}
D.~Hundertmark, P.~C.~Kunstmann, T.~Ried, S.~Vugalter, 
\emph{Cwikel's bound reloaded}, 
preprint, arXiv:1809.05069 (2018).
%
\bibitem{Hundertmark1998}
D.~Hundertmark, E.~H.~Lieb, L.~E.~Thomas, 
\emph{A sharp bound for an eigenvalue moment of the one-dimensional Schr{\"o}dinger operator}, 
Adv. Theor. Math. Phys. \textbf{2}, 719--731 (1998).
%
\bibitem{Hundertmark2000} 
D.~Hundertmark, A.~Laptev, T.~Weidl, 
\emph{New bounds on the Lieb-Thirring constants}, 
Invent. Math. \textbf{140}, 693--704 (2000).
%
\bibitem{Ilyin2022}
A.~A.~Ilyin, A.~Kostianko, S.~Zelik, 
\emph{Applications of the Lieb--Thirring and other bounds for orthonormal systems in mathematical hydrodynamics},
preprint, arXiv:2202.01531 (2022).
%
\bibitem{Keller1961}
J.~Keller, 
\emph{Lower bounds and isoperimetric inequalities for eigenvalues of the Schr{\"o}dinger equation}, 
J. Math. Phys. \textbf{2}, 262--266 (1961).
%
\bibitem{Killip2003}
R.~Killip, B.~Simon, 
\emph{Sum rules for Jacobi matrices and their applications to spectral theory}, 
Ann. of Math. (2) \textbf{158}(1), 253--321 (2003).
%
\bibitem{Killip2009}
R.~Killip, B.~Simon, 
\emph{Sum rules and spectral measures of Schr{\"o}dinger operators with $L^2$ potentials}, 
Ann. of Math. (2) \textbf{170}(2), 739--782, (2009).
%
\bibitem{Laptev1997}
A.~Laptev, 
\emph{Dirichlet and Neumann eigenvalue problems on domains in Euclidean spaces}, 
J. Funct. Anal., \textbf{151}(2), 531--545 (1997).
%
\bibitem{Laptev2012}
A.~Laptev, 
\emph{Spectral inequalities for partial differential equations and their applications}, 
AMS/IP Stud. Adv. Math \textbf{51}, 629--643 (2012).
%
\bibitem{Laptev2021}
A.~Laptev, M.~Loss, L.~Schimmer, 
\emph{On a conjecture by Hundertmark and Simon}, 
preprint, arXiv:2012.13793 (2021)
%
\bibitem{Laptev2021b}
A.~Laptev, L.~Schimmer, 
\emph{A sharp Lieb--Thirring inequality for functional difference operators}, 
SIGMA \textbf{17}, 105, 10 pages (2021). 
%
\bibitem{Laptev2000} 
A.~Laptev, T.~Weidl,
\emph{Sharp Lieb-Thirring inequalities in high dimensions}, 
Acta Mathematica \textbf{184}, 87--111 (2000). 
%
\bibitem{Levin1997}
D.~Levin, M.~Z.~Solomyak, 
\emph{The Rozenblum-Lieb-Cwikel inequality for Markov generators}, 
J. Anal. Math. \textbf{71}, 173--193 (1997).
%
\bibitem{Levitt2014}
A.~Levitt, 
\emph{Best constants in Lieb--Thirring inequalities: a numerical investigation}, 
J. Spectr. Theory \textbf{4}(1), 153--175 (2014).
%
\bibitem{Li1983}
P.~Li, S.~T.~Yau, 
\emph{On the Schr{\"o}dinger equation and the eigenvalue problem}, 
Comm. Math. Phys.\textbf{88}, 309--318 (1983).
%
\bibitem{Lieb1976b}
E.~H.~Lieb, 
\emph{Bounds on the eigenvalues of the Laplace and Schr{\"o}dinger operators}, 
Bull. Amer. Math. Soc. \textbf{82}(5), 751--753 (1976).
%
\bibitem{Lieb1980}
E.~H.~Lieb, 
\emph{The number of bound states of one-body Schr{\"o}dinger operators and the Weyl problem},
In: Geometry of the Laplace operator (Proc. Sympos. Pure Math., Univ. Hawaii, Honolulu, Hawaii, 1979), pp. 241--252, Proc. Sympos. Pure Math., XXXVI, Amer. Math. Soc., Providence, R.I., 1980.
%
\bibitem{Lieb1984}
E.~H.~Lieb, 
\emph{On characteristic exponents in turbulence}, 
Commun. Math. Phys. \textbf{82}, 473--480 (1984).
%
\bibitem{Lieb2010}
E.~H.~Lieb, R.~Seiringer, 
\emph{The stability of matter in quantum mechanics}, 
Cambridge University Press, Cambridge, 2010.
%
\bibitem{Lieb1976}
E.~H.~Lieb, W.~E.~Thirring, 
\emph{Inequalities for the moments of the eigenvalues of the Schr{\"o}dinger hamiltonian and their relation to Sobolev inequalities}, 
In: Studies in Mathematical Physics, pp. 269--303. Princeton University Press, Princeton (1976).
%
\bibitem{Lieb1975}
E.~H.~Lieb, W.~E.~Thirring, 
\emph{Bound for the Kinetic Energy of Fermions which Proves the Stability of Matter}, 
Phys. Rev. Lett. \textbf{35}, 687--689 (1975). Errata ibid., 1116 (1975).
%
\bibitem{Lundholm2013}
D.~Lundholm, J.~P.~Solovej, 
\emph{Hardy and Lieb--Thirring inequalities for anyons}, 
Comm. Math. Phys. \textbf{322}(3), 883--908 (2013)
\bibitem{Nam2020}
P.~T.~Nam, 
\emph{Direct methods to Lieb--Thirring kinetic inequalities}, 
preprint, arXiv:2012.12045 (2020).
%
\bibitem{Poschl1933}
G.~P{\"o}schl, E.~Teller, 
\emph{Bemerkungen zur Quantenmechanik des anharmonischen Oszillators}, 
Z. Phys., \textbf{83}, 143--151 (1933).
%
\bibitem{Roepstorff1994}
G.~Roepstorff, 
\emph{Path integral approach to quantum physics. An introduction}, 
Texts and Monographs in Physics. Springer-Verlag, Berlin, 1994.
%
\bibitem{Rozenblum1972} 
G.~V.~Rozenblum, 
\emph{Distribution of the discrete spectrum of singular differential operators}, 
Dokl. Akad. Nauk SSSR \textbf{202}, 1012--1015 (1972). English translation in Soviet Math. Dokl. \textbf{13}, 245--249 (1972).
%
\bibitem{Rumin2011}
A.~Rumin, 
\emph{Balanced distribution-energy inequalities and related entropy bounds}, 
Duke Math. J. \textbf{160}, 567--597 (2011).
%
\bibitem{Sabin2016}
J.~Sabin, 
\emph{Littlewood--Paley decomposition of operator densities and application to a new proof of the Lieb--Thirring inequality}, 
Math. Phys. Anal. Geom. \textbf{19}, 11 (2016).
%
\bibitem{Simon2005}
B.~Simon, 
\emph{Functional integration and quantum physics}, 
Second edition. AMS Chelsea Publishing, Providence, RI, 2005. 
%
\bibitem{Schimmer2015}
L.~Schimmer, 
\emph{Spectral inequalities for Jacobi operators and related sharp Lieb--Thirring inequalities on the continuum}, 
Comm. Math. Phys. \textbf{334}(1), 473--505 (2015).
%
\bibitem{Schwinger1961}
J.~Schwinger, 
\emph{On the bound states of a given potential}, 
Proc. Nat. Acad. Sci. U.S.A. \textbf{47}, 122--129 (1961).
%
\bibitem{Weidl1996}
T.~Weidl, 
\emph{On the {L}ieb-{T}hirring constants {$L_{\gamma,1}$} for {$\gamma\geq 1/2$}}, 
Comm. Math. Phys. \textbf{178}(1), 135--146 (1996).
%
\bibitem{Zaharov1971}
V.~E.~Zaharov, L.~D.~Faddeev, 
\emph{The Korteweg-de Vries equation is a fully integrable Hamiltonian system}, 
Funkcional. Anal. i Prilov{z}en. \textbf{5}(4), 18--27 (1971), English translation in Funct. Anal. Appl. \textbf{5}, 280--287 (1972).
%
\end{thebibliography}
\end{document}